\documentclass[runningheads,envcountsame]{llncs}
\usepackage[T1]{fontenc}
\usepackage{graphicx}

\usepackage{tikz}
\usetikzlibrary{arrows,shapes,backgrounds,calc,positioning,arrows.meta,patterns,decorations.pathreplacing,patterns.meta,decorations.pathmorphing}

\newcommand{\includeTikzFigure}[1]{%
    \csname #1\endcsname
}

\usepackage[ruled,linesnumbered,noend]{algorithm2e} %
    \SetKwComment{Comment}{$\triangleright$\ }{}
    \SetKwProg{Fn}{function}{$\colon$}{end}
\usepackage{mathtools} %
\usepackage{todonotes} %
\usepackage{xspace}
\usepackage{cite}

\usepackage{amsfonts}
\usepackage{tabularx}

\newcommand{\abs}[1]{{\lvert #1 \rvert }}
\newcommand{\bigO}{\mathcal{O}}
\DeclareMathOperator{\setQ}{\mathbb{Q}}
\DeclareMathOperator{\setN}{\mathbb{N}}

\newcommand{\SPANNER}{\textsc{Spanner}\xspace}

\newcommand{\rmE}{N}

\newcommand{\trivial}{trivial\xspace}
\newcommand{\critical}{critical\xspace}

\newcommand{\decideTimeSymb}{\ensuremath{\mathcal{T}}}

\newcommand{\enc}[1]{\ensuremath{\langle #1 \rangle}}

\newcommand{\maxDeg}{\ensuremath{\Delta}}

\newcommand{\pathSetSymb}{\ensuremath{\mathcal{P}}}
\newcommand{\pathSet}[1]{\ensuremath{\pathSetSymb_{#1}}}

\newcommand{\universe}{\ensuremath{\mathcal{U}}}
\newcommand{\family}{\ensuremath{\mathcal{S}}}

\newcommand{\dist}[3]{\ensuremath{\ell_{#1}[#2,#3]}}

\newcommand{\noF}{\overline{F}}

\newcommand{\oneFunc}{\ensuremath{\mathbf{1}}}

\newcommand{\paramRE}{\ensuremath{\varrho}}%
\newcommand{\paramRW}{\ensuremath{\overline{W}}}%
\newcommand{\paramAP}{\ensuremath{\tau}}%
\newcommand{\paramFES}{\ensuremath{\varphi}}

\newcommand{\paramNTEiS}{\ensuremath{\mu}}%
\newcommand{\paramBW}{\ensuremath{b}}%
\newcommand{\paramBS}{\ensuremath{s}} %
\newcommand{\paramLNS}{\ensuremath{h}}

\newcommand{\bundlewidth}{bundle-breadth\xspace}
\newcommand{\bundlesize}{bundle-size\xspace}
\newcommand{\tightness}{tightness\xspace}

\newcommand{\lmax}{{\ell_{\max}}}

\newcommand{\APSP}{\mathrm{APSP}(n,m)}
\newcommand{\SPSP}{\mathrm{SPSP}(n,m)}
\newcommand{\SSSP}{\mathrm{SSSP}(n,m)}
\newcommand{\APtwoSP}{\mathrm{APSP}_2(n,m)}
\newcommand{\APkSP}{\mathrm{APSP}_k(n,m)}
\newcommand{\APbSP}{\mathrm{APSP}_b(n,m)}

\newcommand{\NP}{\textbf{NP}}
\newcommand{\paraNP}{\textbf{para-NP}}
\newcommand{\Wone}{\textbf{W[1]}}
\newcommand{\Wtwo}{\textbf{W[2]}}

\newcommand{\green}{viable\xspace}

\usepackage{hyperref}
\usepackage{cleveref}
\usepackage{fontawesome5}
\usepackage{multirow}

\newcommand{\linkSymb}{\scalebox{0.7}{\color{red!70!gray}\faExternalLink*}} 
\newcommand{\myLink}[1]{\protect\hyperlink{#1}{\linkSymb}}

\newcounter{paramthm}[theorem]

\newcommand{\defn}[1]{\textcolor{blue!80!black}{\emph{#1}}}

\newcommand{\mySubparagraph}[1]{\smallskip\textbf{#1}}
\spnewtheorem{observation}[theorem]{Observation}{\bfseries}{\itshape}
\crefname{observation}{Observation}{Observations}

\usepackage{orcidlink}
\usepackage{comment}

\begin{document}

\title{Parameterized Complexity of Spanner Problems with Independent Weights and Lengths\thanks{supported by DFG Project \#517835933}}
\titlerunning{Parameterized Complexity of Spanners with Indep.\ Weights and Lengths}
\author{Marius Bächler \orcidlink{0009-0003-2149-1478} \and
Markus Chimani \orcidlink{0000-0002-4681-5550} \and
Henning Jasper \orcidlink{0000-0002-9821-8600}
}
\authorrunning{M. Bächler, M. Chimani, and H. Jasper}
\institute{Institute of Computer Science, Osnabrück University, Germany\\
\email{\{mabaechler, markus.chimani, henning.jasper\}@uos.de}}

\maketitle              %
\begin{abstract}

In this paper, the parameterized complexity of the multiplicative $\alpha$-spanner problem with independent weights and lengths on undirected graphs is considered for the first time.
All prior FPT results (except one on DAGs) assume \emph{basic} instances (i.e., with unit weights and lengths) and are parameterized in the stretch factor $\alpha$ and the (in practice typically non-constant) number of removed edges.

We show that several parameterizations do not allow FPT algorithms.
However, our \emph{exclusion approach} generalizes an existing algorithm for basic instances to arbitrary weights and lengths.
It is parameterized by the total removed weight and a new \emph{tightness} parameter.
The latter is more precise than $\alpha$ and allows us to also improve the best known result for basic instances.
Our second algorithm, called \emph{inclusion approach}, uses the natural parameterization in the spanner's total weight.
We prove that this sole parameter leaves a \Wtwo-hard problem, but also show FPT algorithms exist when augmented with secondary parameters.

\end{abstract}
\section{Introduction}
Let $G=(V,E)$ be a connected undirected graph with edge weights $w \colon E \rightarrow \setN_{\geq 0}$ and edge lengths $\ell \colon E \rightarrow \setN_{\geq 1}$.
For a subgraph $G'\subseteq G$ and any $u,v \in V$, let $\dist{G'}{u}{v}$ be the length of a shortest path from $u$ to $v$ in $G'$, subject to $\ell$.
A subgraph $H=(V,F)$, $F\subseteq E$, is a (multiplicative) \defn{$\alpha$-spanner} for a (typically small constant) \hypertarget{stretch}{\defn{stretch factor}} $\alpha \geq 1$, if it satisfies the \defn{stretch property} $\dist{H}{u}{v} \leq \alpha \dist{G}{u}{v}$ for every \defn{terminal pair} $\{u,v\} \in K = \binom{V}{2}$. 
The stretch property holds for all node pairs even when only enforced for all adjacent pairs $K=E$~\cite{Peleg1989}.
The natural optimization problem is to find an $\alpha$-spanner of minimum weight $w(H) \coloneqq \sum_{e \in F} w(e)$. This gives rise to the canonical decision problem:
\begin{definition}[\SPANNER problem]
    Given an instance $I=(G,w,\ell,\alpha,W)$, with \hypertarget{weight}{\defn{weight bound}} $W\in\mathbb{N}$, is there an $\alpha$-spanner $H=(V,F)$ with $w(H)\leq W$?
\end{definition}
In many scenarios, 
$\alpha$ is not considered part of the input but a constant part of the problem definition.
See~\cite{Ahmed2020} for an overview on spanner literature.

Spanners in general graphs were first introduced in the context of communication networks~\cite{Peleg1989}.
They since found applications in various fields including approximate distance oracles, graph drawing, and access control hierarchies, see, e.g.,~\cite{Ahmed2020,wir2026GMSIP} for details. 
In practice, the weight and length of an edge are often entirely independent, we say \defn{decoupled}.
However, the literature often considers simplifying assumptions: 
In \defn{coupled} %
instances, each edge's weight equals its length, i.e., $ w(e) = \ell(e) > 0$, for all $e \in E$.
\defn{Unit-length} or \defn{unit-weight} instances assume $\ell(e) = 1$ or $w(e) = 1$ for all $e\in E$, respectively.
In \defn{basic} %
instances, we even assume %
$w(e)=\ell(e) = 1$ for all $e \in E$.
There are some exact ILP approaches and several approximation algorithms, see~\cite{Ahmed2020} for an overview and~\cite{chimani2022,jasper2024} for practical evaluations. %
However, only few approximations have guarantees for the decoupled setting~\cite{wir2025simple,chimani2014network,grigorescu2023approximation}, see also~\cite{wir2026GMSIP}.

Let $n\coloneqq |V|$ and $m\coloneqq |E|$.
Already the basic \SPANNER problem is \NP-hard for constant $\alpha \geq 2$~\cite{Peleg1989}, even on severely restricted graph classes~\cite{Gomez2023}, and hard to approximate~\cite{Kortsarz2001,Elkin2007}.
This motivates the consideration of \emph{fixed parameter tractable} (FPT) algorithms.
An FPT algorithm with parameter $k$ runs in $\bigO(f(k) \cdot p(\enc{I})$ time, where $f(k)$ is any computable function and $p(\enc{I})$ is a polynomial only dependent on the encoding length of the instance, see \cite{Downey1995} for %
details.
Note that $k$ may be a \defn{combined parameter}, i.e., an amalgamation of  several distinct parameters. A \defn{witness}---in our case a feasible spanner adhering to the restrictions given by the parameters---proves that a given instance is indeed a yes-instance.

Observe that $G$ is always a spanner w.r.t.\ itself and the spanner-property is supergraph-monotone. %
For basic \SPANNER instances, Kobayashi~\cite{Kobayashi2018} presents an FPT algorithm that is parameterized in $\alpha$ and the number $\paramRE$ of edges \emph{removed} from $G$ to obtain $H$ (i.e., we require $|E\setminus F|=\paramRE$).
The running time is 
$\bigO(\paramRE^{2\paramRE+3} \alpha^{\paramRE+1} (\alpha+1)^{\paramRE+1}+ nm)$.
In~\cite{Kobayashi2019}, they extend their work to basic linear $(\alpha, \beta)$-spanners and basic additive $+\beta$-spanners, which require $\dist{H}{u}{v} \leq \alpha \dist{G}{u}{v} + \beta$ and $\dist{H}{u}{v} \leq \dist{G}{u}{v} + \beta$, for all $\{ u, v \} \in \binom{V}{2}$, respectively.
Fomin et al.~\cite{Fomin2022} show that on \emph{directed} graphs, the basic \SPANNER problem parameterized in $\alpha$ and $\paramRE$ can be solved in  $(4\alpha)^{\paramRE} \cdot n^{\bigO(1)}$ randomized time.
For directed unit-weight spanners on directed \emph{acyclic} graphs (DAGs), they bound the running time by $\paramRE^{2\paramRE} \cdot n^{\mathcal{O}(1)}$.
They show that finding directed basic $+\beta$-spanners parameterized in $\paramRE$ is \Wone-hard even on DAGs for every fixed~$\beta \geq 1$.
Fluschnik~\cite{fluschnik2026} considers a variant of the \SPANNER problem in the context of safe bicycle network design (SBND), where bike-friendly edges can be seen as weight-$0$, but stretch requirements are only given for a specific subset $K\subseteq \binom{V}2$.
SBND is \paraNP-hard w.r.t.\ several (combined) parameters, typically using the fact that solutions do not need to establish the stretch property for all node-pairs. They propose two FPT algorithms, one w.r.t.\ $|K|$ and the feedback edge set number, the other w.r.t.\ $|E'|$ (where $E'$ are the non-bike-friendly edges).
In the context of spanners, both parameters would imply a constant instance size (given a straight-forward polynomial kernelization on the bike-friendly edges to obtain a complete subgraph on $V(E')$).

There are no FPT algorithms for non-basic instances, except for unit-weight DAGs~\cite{Fomin2022}. %
Moreover, all algorithms are parameterized in the number of removed edges $\paramRE$, which is impractical: spanners are typically used to sparsify a (possibly even asymptotically) much denser graph $G$, contradicting the assumption of a constant $\paramRE$. 
There are no FPT algorithms for the \emph{natural} parameterization, i.e., the objective function value $W$.
For unit-weight instances, this is the spanner's \defn{size} $|F| \geq n-1$.
Yet, this still leaves room for a derived parameter that counts the (possibly constant) number of spanner edges beyond such a linear lower bound.
We can leverage ``trivial'' edges---those that are contained in every optimum spanner---and parameterize by the number $\paramNTEiS$ of ``nontrivial'' spanner edges. 
A key example is the \defn{\textsc{SpannerAugmentation}} problem~\cite{elin2001client,Elkin2007}, where a fixed subgraph (e.g., spanning tree) is given and we add edges to establish the spanner property.
This is a special case of the decoupled \SPANNER problem, modeling fixed edges with weight~$0$.
The natural objective parameterization remains sensible: the spanner's weight may be constant despite having $\Omega(n)$ edges.

\newcommand{\critsym}{$\forall C$}
\begin{table}[t]
    \caption{
    Overview on parameterizations and results.
    $\APSP,\APtwoSP$ and $\SPSP$ bound the running times for all-pairs-, all-pairs-2-, and single-pair-shortest-path computations, respectively.
    By \Cref{cor:W}, the last four results (marked with \raisebox{-0.5mm}{$^\dagger$}) also work using parameter $W$ instead of~$\paramNTEiS$.
    }    
    \label{tab:results}
    \centering

                \resizebox{\textwidth}{!}{
\begin{tabular}{|c@{\,}c|c|p{13cm}|c|}
\hline
&& \textbf{symbol} &\textbf{ parameter description} & \textbf{link} \\
\hline
\multirow{7}{*}{\rotatebox{90}{\textbf{instance}}} &
\multirow{7}{*}{\rotatebox{90}{\textbf{dependent}}} &
$\alpha$ & stretch factor & \myLink{stretch}\\
&&$\paramAP$ &  \emph{\tightness}: max.\ \# of nontrivial edges along the tightest alternative path (\critsym) & \myLink{paramAP} \\
&&$\paramFES$ & feedback edge set number &  \myLink{paramFES} \\
&&$\paramBW$ & \emph{\bundlewidth}: max.\ \# of settling paths (\critsym) & \myLink{paramBW} \\
&&$\paramBS$ & \emph{\bundlesize}: max.\ \# of nontrivial edges on settling paths (\critsym) & \myLink{paramBS} \\
&&$\paramLNS$ & \emph{local neighborhood size}: max.\ \# of nodes on settling paths (\critsym) & \myLink{paramLNS} \\
&&$\maxDeg$ & max.\ node-degree of $G$ &  \textbf{---}\\ %
\hline
\multirow{4}{*}{\rotatebox{90}{\textbf{solution}}} &
\multirow{4}{*}{\rotatebox{90}{\textbf{depend.}}} &
$\paramRE$ & (min) \# of nontrivial edges removed from $G$ & \textbf{---}\\ %
&&$\paramRW$ & (min)\ total weight removed from $G$  & \textbf{---}\\ %
&&$\paramNTEiS$ & (max) \# nontrivial edges in $H$ &  \textbf{---}\\ %
&&$W$ & (max) total weight of $H$  & \textbf{---}\\ %
\hline
\multicolumn{5}{r}{(\critsym = over all critical terminal pairs)}
\end{tabular}
    }

\medskip
    \resizebox{\textwidth}{!}{
    \begin{tabular}{|lll|c|c|c|c|}
    \hline
        \multicolumn{3}{|c|}{\textbf{parameter(s)}} & \textbf{variant} & \textbf{known result} &  \textbf{our contribution}& \textbf{ref.}\ \\
        \hline
        \hline
        & \multicolumn{2}{l|}{$\paramAP$ or $\alpha$} & basic & \paraNP-hard~\cite{Peleg1989,Gomez2023} & \textbf{---} & \textbf{---}\\
        \hline
        $\paramRW$ & $\paramAP$ & & decoupled & \textbf{---} &
        $\bigO\big( \paramRW^{2\paramRW+1} \paramAP^{2\paramRW} \cdot \SPSP + \APtwoSP \big)$ & Thm.\ \ref{thm:exclude} \\
        $\paramRW$ & $\alpha$ & & coupled & \textbf{---} & $\bigO\big( \paramRW^{4\paramRW+1} \alpha^{2\paramRW} \cdot \SPSP+ \APtwoSP \big)$ & Cor.\ \ref{cor:exCoupled} \\
        $\paramRW$ & $\alpha$ & & unit-length & \textbf{---} & $\bigO\big( \paramRW^{2\paramRW+1} \alpha^{2\paramRW}\cdot m +nm \big)$ & Cor.\ \ref{cor:unitlengthex} \\
        $\paramRE$ & $\paramAP$ & & unit-weight & \textbf{---} & $\bigO\big(\paramRE^{2\paramRE+1} \paramAP^{2\paramRE} \cdot \SPSP+ \APtwoSP \big)$ & Cor.\ \ref{cor:exUnitWeight} \\
        $\paramFES$ & $\paramAP$ & & unit-weight & \textbf{---} & $\bigO\big(\paramFES^{2\paramFES+2} \paramAP^{2\paramFES} \cdot \SPSP+ \APtwoSP \big)$ & Cor.\ \ref{cor:exFES} \\
        $\paramRE$ & $\alpha$ & & basic & $\bigO\big(\paramRE^{2\paramRE+3}\alpha^{\paramRE+1}(\alpha+1)^{\paramRE+1} + nm\big)$~\cite{Kobayashi2018} & 
        $\bigO\big( \paramRE^{2\paramRE+3} \alpha^{2\paramRE+2}+nm \big)$ & Cor.\ \ref{cor:exBasic}\\

        \hline
        \hline
        \multicolumn{2}{|l}{$\paramNTEiS$ or $W$}  & & unit-length & \textbf{---} & \Wtwo-hard & Thm.\ \ref{thm:W2hard} \\

         & \multicolumn{2}{l|}{$\paramBW$ or $\paramBS$ or $\paramLNS$} & basic & \textbf{---} & \paraNP-hard & Obs.\ \ref{obs:unitlength} \\
         
         \hline
         $\paramNTEiS$ & $\paramBW$ & & decoupled & \textbf{---} & $\bigO\big(\paramBW^{\paramNTEiS} \cdot \APSP + \APbSP\big)$ & Thm.\ \ref{thm:f-bw} $^\dagger$ \\

        $\paramNTEiS$ & $\paramBS$ & & decoupled & \textbf{---} & $\bigO\big(\paramBS^{\paramNTEiS} \cdot \APSP + m^3+m^2 n \log n\big)$ & Thm.\ \ref{thm:f-bs} $^\dagger$\\

         $\paramNTEiS$ & $\paramLNS$ & & decoupled & \textbf{---} & $\bigO\big(\paramLNS^{2\paramNTEiS} \cdot \APSP + m^3+m^2 n \log n\big)$ & Cor.\ \ref{cor:lns} $^\dagger$\\

        $\paramNTEiS$ & $\alpha$ & $\maxDeg$ & unit-length & \textbf{---}  & $\bigO\left(\maxDeg^{\alpha \paramNTEiS} \cdot nm+ \maxDeg^\alpha m^2\right)$ & Cor.\ \ref{cor:ulengthfpt} $^{\dagger}$ \\
        \hline
    \end{tabular}
    }

\end{table}

\mySubparagraph{Contribution.} 
We present the first FPT results for the decoupled \SPANNER problem.
Our first parametrization (\Cref{section:alg_exclusion}) is a generalization of Kobayashi's parametrization for basic \SPANNER instances~\cite{Kobayashi2018}.
The central parameter is the minimum weight $\paramRW$ to be \emph{removed} from $G$ to yield~$H$, combined with a parameter $\paramAP$ bounding the number of nontrivial edges on an alternative path (formally defined later). For the basic scenario, $(\paramRW,\paramAP)$ becomes $(\paramRE,\alpha)$ as considered in~\cite{Kobayashi2018}, where we improve upon their running time.
Our second parameterization (\Cref{section:alg_inclusion}) considers the number of nontrivial edges $\paramNTEiS$ to be \emph{contained} in the spanner, combined with a locality parameter (\bundlewidth~$\paramBW$, \bundlesize~$\paramBS$, or local neighborhood size~$\paramLNS$, see later for details).

Further, we show that fixing the parameters only individually still leaves hard problems.
Combined parameters are thus indeed unavoidable.
For both algorithms, we show that in certain scenarios (e.g. coupled instances) their running time can be improved and/or they can be understood to be parameterized by simpler parameters. See \Cref{tab:results} for an overview on our parameters and results.

\section{Preliminaries}
\label{section:pre}
For any edge set $X \subseteq E$ and function $f\colon E \to \setQ$, we denote $f(X) \coloneqq \sum_{e \in X} f(e)$ and $f_{\min}(X) \coloneqq \min_{e \in X}(f(e))$.
Similarly, for any path $P$ and function $f$, we denote $f(P) \coloneqq f(E(P))$.
We denote the constant function $\oneFunc \colon E \to \{1\}$.
For a terminal pair $\{u,v\} \in K=E$, a $uv$-path $P$ is a simple path that goes from $u$ to $v$.
It is \defn{settling} if it satisfies the stretch property for $\{u,v\}$, i.e., $\ell(P) \leq \alpha \cdot \dist{G}{u}{v}$. 
Observe that, in general, the edge $\{u,v\}$ may not always form a settling $uv$-path.
We only have to actively enforce the stretch property for all \emph{metric} edges $\{u,v\} \in E$ with $\ell(\{u,v\}) = \dist{G}{u}{v}$. If a nonmetric edge does not even settle itself, we can remove the edge from the instance as a preprocessing step.
$P$ is an \defn{alternative path} if it does not contain the edge $\{u,v\}$ itself.
The set of all settling paths is denoted by $\pathSet{uv}$.
If $\pathSet{uv}$ contains no alternative path, then the edge $\{u,v\}$ is
\defn{mandatory} and part of every feasible spanner.
\SPANNER instances typically contain many \hypertarget{triv}{\defn{\trivial}} edges $T \subseteq E$, i.e., edges that w.l.o.g.\ can always be assumed to be part of an optimum solution.
Besides mandatory edges, these include all edges $e \in E$ with $w(e)=0$.
The latter numerously occur in \textsc{SpannerAugmentation} problems (see above).
Further, studies of exact algorithms for coupled \SPANNER instances suggest that, especially for small $\alpha$, many edges are mandatory in practice~\cite{jasper2024}.
The inclusion of all trivial edges already settles some terminal pairs, i.e., those connected by a settling path that only consists of trivial edges.
Terminal pairs $C \subseteq K=E$ not settled by trivial edges are called \defn{\critical} and must be settled by the inclusion of \defn{nontrivial} edges $N \coloneqq E\setminus T$. 
Observe that $C\subseteq N$, as every $\{u,v\} \in T$ settles itself. See Fig.\ \ref{fig:paramEG} for a small  example showcasing these and subsequent definitions.

\begin{figure}[tb]
    \centering
    \resizebox{1.2\textwidth}{!}{\includeTikzFigure{paramEG}}
    \caption{Example for definitions, with $\alpha=2$. Edges have length and weight $1$, except that the green edge has weight $0$. Trivial edges $T$ are thick (thick black are even mandatory). 
    By symmetry, the values $\paramAP,\paramLNS,\paramBS,\paramBW$ are induced by both elements of $C$; we list them explicitly for $\{x,y\}\in C$.
    The red shading shows the two paths in $\mathcal{P}_{xy}$.
} %
    \label{fig:paramEG}
\end{figure}

\mySubparagraph{Shortest Paths.}
We can always solve the single-source-shortest-path (SSSP)---and thus single-pair-shortest-path (SPSP)---problem using Dijkstra's $\bigO(m+n\log n)$ time algorithm.
The \defn{all-pairs-shortest-path (APSP)} problem is solvable in $\bigO(nm+n^2\log n)$ time by $n$ SSSP calls.
Improved running times are possible for specific graph densities~\cite{duan2025breaking} or length functions; e.g., unit-lengths allow APSP in $\bigO(nm)$ time via breadth-first-search (BFS).
Repeatedly solving similar instances may allow faster running times via \emph{dynamic} shortest paths~\cite{demetrescu2004new,garg2021dynamizing}.
We use the monikers $\SSSP$, $\SPSP$, and $\APSP$ to denote the running time of suitable algorithms on a graph with $n$ nodes and $m$ edges.

For $k \in \mathbb{N}$, we use the algorithm of~\cite{katoh1982efficient} to solve the single-pair-$k$-shortest-path problem in $\bigO(k\cdot \SSSP)$ time, and consequently, the \defn{all-pairs-$k$-shortest-paths (AP$k$SP)} problem in $\APkSP \in \bigO(n^2 k\cdot \SSSP)$ time.
To detect mandatory edges, we use these algorithms with $k=2$ and check for which terminal pairs the second path is still settling.
As all weight-$0$ edges can be found in time $m$, all trivial edges $T \subseteq E$ can be found in $\APtwoSP \in \bigO(n^2\cdot \SSSP)$ time.
For unit-lengths instances, $\APtwoSP \in \bigO(nm)$ by checking subtree-linking edges in $n$-many BFS-trees, see, e.g.,~\cite{Kobayashi2018}.

\mySubparagraph{Independent Set.}
In \Cref{section:alg_exclusion}, we will require a fact about independent sets that seems not to be readily available in known literature. Recall that, given a graph $G$, an independent set $U\subseteq V$ is a node subset such that no two elements of $U$ are adjacent. See Appendix~\ref{apx:MissingProofs} for the full proof.
\begin{lemma}\label{lem:IS}
    Let $G=(V,E)$ be an $n$-node graph with $|V(G[V'])| \geq |E(G[V'])|$ for all $V' \subseteq V$, where $G[V'] \subseteq G$ is the node-induced subgraph. Then $G$ has an independent set of size at least $\lfloor \frac{n}{2}\rfloor$, which can be found in $\bigO(n)$ time.
\end{lemma}

\section{Exclusion Algorithm}\label{section:alg_exclusion}
For a given spanner $H=(V,F)$ of a graph $G=(V,E)$, let $\noF\coloneqq E\setminus F$ be the edges \emph{removed} from $G$ to yield $H$.
In~\cite{Kobayashi2018}, Kobayashi presented an FPT algorithm for the \emph{basic} \SPANNER (i.e., unit weights and unit lengths) problem w.r.t.\ the combined parameter $(\paramRE',\alpha)$ where $\paramRE'$ is a lower bound on $|\noF|$; i.e., we ask whether there exists a spanner where at least $\paramRE'$ edges are removed. Originally, the algorithm was described as parameterized by the single parameter $\paramRE'$ but assuming constant $\alpha$; indeed the stronger result as a combined parameter holds.

We generalize the FPT algorithm to the decoupled scenario, using the combined parameter $(\paramRE,\paramAP)$: 
Here, \hypertarget{paramRE}{$\paramRE\leq \paramRE'$} is the number of \emph{nontrivial} edges removed from $G$\footnote{Indeed, the algorithm~\cite{Kobayashi2018} could, without any real effort, also be seen as being parameterized in $\paramRE$ instead of the formally slightly weaker $\paramRE'$.}.
For each critical terminal pair $\{u,v\}\in C$, let the \defn{tightest} alternative path $A_{uv}\in\pathSet{uv}$ be an alternative path with the minimum number of nontrivial edges. Let $A^N_{uv}\coloneqq A_{uv}\cap N$ denote its nontrivial edges. %
The \hypertarget{paramAP}{\defn{\tightness}} $\paramAP:=\max\{ |A^N_{uv}| : \{u,v\}\in C\}$, is the maximum number of nontrivial edges over any such path. 
For all $\{u,v\} \in C$, $|A^N_{uv}| \geq 1$ as otherwise $\{u,v\} \notin C$.

We argue that the \tightness $\paramAP$ is a ``better'' parameter than $\alpha$ as it precisely captures $\alpha$'s role in the algorithm.
On unit-length (and thus also basic) \SPANNER instances we have $\paramAP\leq\alpha$.
As Kobayashi's algorithm~\cite{Kobayashi2018} hinges on this property, it requires the parameterization (or constant) $\alpha$.
By making $\paramAP$ explicit, we are able to improve on their result.
Further, we are able to handle instances with non-constant $\alpha$ given that the graph is dense enough to allow a constant tightness.
\Cref{obs:largealpha} showcases such instances that are indeed hard.

For each of the at most $m$-many $\{u,v\}\in C$, computing the tightest alternative path $A_{uv}$ is a special case of the \NP-hard constrained-shortest-path (CSP) problem~\cite{Garey1979}:
Given the graph $G'=(V,E'\coloneqq E \setminus \{ \{u,v\} \})$ with, for each $e \in E'$, lengths $\ell(e)$ and binary-costs $c(e)=1$ if $e \in N$ and $c(e)=0$ otherwise.
The goal is to determine a minimum-cost $uv$-path of length at most $\alpha \dist{G}{u}{v}$.
For any $\varepsilon>0$, the fully-polynomial-time approximation scheme for CSP in~\cite{lorenz2001simple} finds a sufficiently short path of cost at most $(1+\varepsilon)$ the optimum in $\mathcal{O}((\log \log n + 1/ \varepsilon)nm)$ time.
Since, for any alternative $uv$-path $A$ we have $1 \leq c(A) \leq n-1$, setting $\varepsilon= \frac{1}{n+1}$ yields an optimum solution in $\bigO(n^2 m)$ time. Hence, we can compute $\paramAP$ in $\bigO(n^2 m^2)$ time.

Kobayashi's algorithm~\cite{Kobayashi2018} computes a \emph{kernel}, i.e., an answer-equivalent problem instance of size only dependent on the parameters. As they consider basic instances $(G=(V,E),w=\oneFunc,\ell=\oneFunc,\alpha,W)$, the only trivial edges are those with no alternative path.
They show that by simply removing these edges from $G$, one obtains an equivalent instance $( G'=(V,N),w'=\oneFunc,\ell'=\oneFunc,\alpha,W)$ that allows a spanner $H'\subseteq G'$ with $W=|E(H')|=|N|-\paramRE$ if and only if $H'+T$ is a spanner in $G$ with $W=|E(H'+T)|=|E|-\paramRE$. They then show that either $|N|$ is bounded by a function of $(\paramRE,\alpha)$ (we have a kernel), or the answer is YES.
However, for decoupled \SPANNER instances, this kernelization step is no longer correct.
We show that the core ideas can be adapted using more finely tuned arguments.

In the unit-weight case, the parameter $\paramRE$ is the complement of the objective function:
Consider the largest $\paramRE$ for which an instance parameterized by $\paramRE$ is a yes-instances; it allows for no spanner with fewer than $|F|=|E|-\paramRE$ edges.
For the weighted case, we analogously parameterize in \hypertarget{paramRW}{$\paramRW$}, i.e., the required total weight of the removed edges.
Clearly, for unit-weights $\paramRW = \paramRE$.
We describe our algorithm using \paramRW, understanding that for unit-weights this reduces to \paramRE.

\mySubparagraph{Algorithm parameterized by $(\paramRW,\paramAP)$.}
Since we never remove trivial edges, we can focus on nontrivial edges $N$.
Similar to~\cite{Kobayashi2018}, our algorithm builds on a case distinction on $|N|$:
if it is small, i.e., $\abs{N} \leq \paramAP^2\paramRW^2$, 
we simply enumerate all feasible spanners. 
Otherwise, we always answer YES, since we can remove edges of sufficient total weight while maintaining a feasible spanner.

Extracting a spanner in the latter case requires some work; see \Cref{alg:findWitness}. 
For each $f=\{u,v\}\in C$, our witness either contains $f$ or its tightest alternative path. 
When removing any $e \in E$, the spanner not only needs to retain $A_e$, but also every $e' \in E$ whose tightest alternative path $A_{e'}$ contains~$e$.
We always start by computing $A_{uv}$ for each critical pair $\{u,v\} \in C$. 
Then, depending on whether $\paramAP=1$ or $\paramAP\geq 2$, we perform vastly different algorithmic methods.
To improve running times for basic instances, we use a tighter bound for $|N|$ (compared to~\cite{Kobayashi2018}).
This requires a more intrinsic argument to show the existence of a witness. 
Note that for basic instances $\paramAP=1$ does not arise as it resembles $\alpha=1$. %

\newcommand{\auxG}{{\tilde G}}
Consider the case $\paramAP=1$. We construct an auxiliary conflict graph $\auxG\coloneqq (N, \bigcup_{\{u,v\}\in C, e\in A^N_{uv}} \{ \{u,v\}, e \})$ that uses the nontrivial edges as its node set;
for each critical pair $\{u,v\}\in C$ and its nontrivial edge $e\in A^N_{uv}$ (note that $|A^N_{uv}|=\paramAP=1$), we join them by an edge in $\auxG$. Each node in $\auxG$ is responsible for at most one edge in $\auxG$, while a single edge in $\auxG$ exists due to either one or two nodes of $\auxG$. Thus, $\auxG$ satisfies the prerequisites of \Cref{lem:IS}. Consider obtaining a subgraph $H$ of $G$ by removing some edge $e\in N$. By retaining all edges in $H$ that are adjacent to $e$ in $\auxG$, we guarantee that $H$ is a feasible spanner. In other words, we can remove an independent set in $\auxG$ from $G$ and guarantee that the resulting graph is a feasible spanner.
By \Cref{lem:IS}, we can find an independent set of  size at least $\lfloor\frac{|N|}{2}\rfloor$ in~$\auxG$. 
Since $|N|\geq\paramRW^2+1$, we thus remove at least $\lfloor \frac{\paramRW^2 +1}{2} \rfloor = \lceil\frac{\paramRW^2}{2}\rceil\geq \paramRW$ many nontrivial edges, each of which has weight at least~$1$.

Now, consider the case $\paramAP\geq 2$.
We construct the spanner by iteratively removing edges from the original graph. 
To this end, we first only decide on \emph{removal candidates} $R_i \subseteq N$, which in turn determine \emph{locked edges} $L_i \subseteq E$.
Let $L_0 \coloneqq \emptyset$. For $i=1,\ldots,\paramRW$, let $R_i$ be an arbitrary set of  $((\paramRW-i)\paramAP+1)$-many 
edges from $N \setminus L_{i-1}$ (we later prove its existence).
We set $L_i \coloneqq L_{i-1} \cup \{ A^N_{uv} \cup \{ \{u,v\} \} \mid \{u,v\} \in R_i\}$.
Then, \emph{in reverse order} $i=\paramRW,\ldots,1$, we pick an edge between a critical terminal pair $\{u_i,v_i\} \in R_i \setminus \bigcup_{j>i} A^N_{u_jv_j}$, i.e., a nontrivial candidate edge of $R_i$ that does not lie on a shortest alternative path of any of the previously picked edges.
Such a pick is always possible, since $|R_i|=(\paramRW-i)\paramAP+1 > \sum_{j>i}\paramAP\geq \sum_{j>i} |A^N_{u_jv_j}|\geq \big|\bigcup_{j>i} A^N_{u_jv_j}\big|$ is large enough. 
By removing all so-picked edges, we yield a spanner that is feasible and removes at least total weight $\paramRW$.

\begin{theorem}\label{thm:exclude}
    The above algorithm (see also Appendix~\ref{apx:Ex}) is an FPT algorithm that decides the \SPANNER problem parameterized by the total removed weight \paramRW and the instance's \tightness $\paramAP$ in 
    $\decideTimeSymb\in\bigO\big( \paramRW^{2\paramRW+1} \paramAP^{2\paramRW} \cdot \SPSP + \APtwoSP \big)$
    time and identifies a witness (if it exists) in $\bigO(\decideTimeSymb+n^2m^2)$ time.
    If $\paramAP$ is not given, we can compute it in $\bigO(n^2m^2)$ time.
\end{theorem}
\begin{proof}
The running time of computing $\paramAP$ has already been discussed above.
Identifying the (non)trivial edges takes $\APtwoSP$ time (cf.\ \Cref{section:pre}).
First, we consider the case $|N| \leq \paramAP^2\paramRW^2$.
Since $w(e)\geq 1$ for all $e \in N$, at most $\paramRW$-many edges must be removed.
Hence, there are at most $\sum_{i=1}^{\paramRW} \binom{|\rmE|}{i} \in \bigO(|\rmE|^{\paramRW}) \subseteq \bigO( (\paramAP^2\paramRW^2)^{\paramRW})$-many subgraphs to check for feasibility.
Each check requires $\bigO(\paramRW \cdot \SPSP)$ time.
The total running time follows.

Next, we prove that we always have a yes-instance if $|N|>\paramAP^2\paramRW^2$. We do so by showing that our algorithm always finds a witness. 
For $\paramAP=1$, we already argued this above.
For $\paramAP \geq 2$, the correctness hinges on the claim that we can always find large enough removal candidate sets $R_i$ with $|R_i|=(\paramRW-i)\paramAP+1=\paramAP\paramRW+1-i\paramAP$. 
Observe that the sets $R_i$ are pairwise disjoint.
Each $L_{i}$ consists of all nontrivial edges $e$ for which, for some $j\leq i$, either $e\in R_j$ or there exists an edge $\{u,v\}$ with $e\in A^N_{uv}$. The sets $L_i$ form an increasing (nested) sequence.
Crucially, an edge $e\in R_i$, for any $i$, cannot already be in $L_{i-1}$. Hence $|L_{i}|\geq|L_{i-1}|+|R_i|$. 
By this strict monotonicity, we only need to argue that even in the last iteration, there are still enough edges in $N\setminus L_{\paramRW-1}$ to select the singleton $R_{\paramRW}$, i.e., we show that $L_{\paramRW-1}\leq \paramAP^2\paramRW^2<|N|$, even if every candidate edge of every preceding set $R_j$ ($1\leq j<\paramRW$) contributes the maximum of $\paramAP+1$ distinct edges to $L_{\paramRW-1}$. 
Observe that $\sum_{j=1}^{\paramRW-1} |R_j| = \sum_{j=1}^{\paramRW-1} (\paramAP\paramRW+1-j\paramAP)
  = (\paramRW-1)(\paramAP\paramRW+1)-\paramAP\sum_{j=1}^{\paramRW-1}j = (\paramAP\paramRW^2-\paramAP\paramRW+\paramRW-1)-\frac{\paramAP}{2}(\paramRW^2-\paramRW)
  = \frac12\paramAP\paramRW^2-\frac12\paramAP\paramRW+\paramRW-1$.
With this and the fact $\paramAP\geq2$, we indeed have
\begin{align*}
|L_{\paramRW-1}| &\leq \bigg(\frac12\paramAP\paramRW^2-\frac12\paramAP\paramRW+\paramRW-1\bigg)(\paramAP+1)\\
&= \frac12\paramAP^2\paramRW^2\underbrace{-\frac12\paramAP^2\paramRW+\frac12\paramAP\paramRW}_{\leq 0}\ \underbrace{\vphantom{\frac12}-\paramAP}_{\leq0} \ \underbrace{+ \frac12\paramAP\paramRW^2}_{\leq \frac14\paramAP^2\paramRW^2}\underbrace{+\vphantom{\frac12}\paramRW}_{\leq \frac14\paramAP^2\paramRW^2}\underbrace{-\vphantom{\frac12}1}_{\leq0} \leq \paramAP^2\paramRW^2.%
\end{align*}

Consider the time complexity of the witness algorithm.
First, for each $\{u,v\}\in C$, it computes the tightest alternative path $A_{uv}$ in $\bigO(mn^2)$ time (see above).

For $\paramAP=1$, we just construct $\auxG$ and find the independent set, which only requires linear time.
For $\paramAP\geq 2$, we perform $\paramRW$-many iterations to establish $R_i$ and $L_i$, for $1\leq i\leq\paramRW$. Using a temporary list for the elements of $N$ from which the edges are removed when adding them to $R_i$ or $L_i$, allows us to identify $R_i$ in time  $\bigO(|R_i|)$ and the \emph{new} locked elements in $L_i$ in $\bigO(\paramAP|R_i|)$ time.
Thus, over all iterations, we require at most $\bigO(|L_i|)\subseteq\bigO (\paramAP^2\paramRW^2)$ time for the first loop. Picking an item at any (of the $\paramRW$-many) iterations of the second loop requires (only) $\bigO(\paramAP)$ time: pick any element and remove the edges of $A^N_{uv}$ from the candidate sets with smaller index. Thus, besides the computation of the tightest alternative paths, the running time for the case $N\leq\paramAP^2\paramRW^2$ dominates.
\qed \end{proof}

\begin{algorithm}[tb]
\DontPrintSemicolon
    \caption{Find witness if $|N|\,{>}\,\paramAP^2 \paramRW^2$ for instance $(G{=}(V,E),w,\ell,\alpha,W)$, \rlap{parameterized by $(\paramRW,\paramAP)$}\label{alg:findWitness}}

        $A_{uv}$ $\gets$ tightest alternative path for each $\{u,v\} \in C$\\
        \If{\paramAP=1}{
            construct conflict graph $\auxG\coloneqq (N, \bigcup_{\{u,v\}\in C, e\in A^N_{uv}} \{ \{u,v\}, e \})$\\
            find independent set $\noF$ in $\auxG$ using \Cref{lem:IS}\\
            \Return $(V,E\setminus \noF)$ \Comment*[r]{feasible spanner of suitable weight}
            }
        $L_0 \gets \emptyset$, $\noF \gets \emptyset$\\
        \For{$i = 1, \dots, \paramRW$}{
            $R_i$ $\gets$ any subset of $N \setminus L_{i-1}$ with $|R_i|=\paramAP\paramRW$\\
            $L_i \gets L_{i-1} \cup \{A^N_{uv} \cup \{\{u,v\}\} \mid \{u,v\} \in R_i \cap C\}$
        }
        \For{$i = \paramRW, \dots, 1$}{
            $\{u_i,v_i\}$ $\gets$ any edge in $R_i \setminus \bigcup_{j>i} A^N_{u_jv_j}$\\
            add $\{u_i,v_i\}$ to $\noF$
        }
        \Return $(V,E\setminus\noF)$ \Comment*[r]{feasible spanner of suitable weight}
    
\end{algorithm}

If given some upper bound $\lmax$ on the edge lengths, we have $\paramAP\leq \alpha \lmax$. 
Then we can compute all $A_{uv}$ via a simple APSP, 
where a $uv$-path must not use $\{u,v\}$ itself, as any alternative path, even if non-tightest, suffices. 
We can compute this in $\bigO(\APtwoSP) \subseteq \decideTimeSymb$ time. %
On coupled instances, we can w.l.o.g.\ further assume that $w(e)=\ell(e) < \paramRW$ for all $e \in N$, as otherwise removing a single edge violating this property would yield a spanner with sufficient weight removed.

\begin{corollary}\label{cor:exCoupled}
    The above algorithm is an FPT algorithm w.r.t.\ $(\paramRW,\lmax,\alpha)$ (setting $\paramAP=\alpha \lmax$). For coupled instances, it is an FPT algorithm
    w.r.t.\ $(\paramRW,\alpha)$ (setting $\paramAP=\alpha \paramRW$).
    The running time for identifying a witness improves %
    to $\bigO(\decideTimeSymb)$. 
\end{corollary}

For unit lengths, not only do we have $\lmax=1$, but also all shortest path computations simplify to linear-time BFS-traversals as discussed in \Cref{section:pre}.

\begin{corollary}\label{cor:unitlengthex}
    For unit-length instances, the above algorithm is FPT w.r.t.\ $(\paramRW,\alpha)$ and decides the problem (incl.\ a witness) in $\bigO\big( \paramRW^{2\paramRW+1} \alpha^{2\paramRW}\cdot m +nm \big)$ time.
\end{corollary}

\begin{corollary}\label{cor:exUnitWeight}
    For unit-weight instances, we can substitute the parameter $\paramRW$ in the above algorithm by the number of removed nontrivial edges $\paramRE$. For basic instances, using \Cref{cor:unitlengthex}, we yield a parameter dependency only on $(\paramRE,\alpha)$.
\end{corollary}

Incidentally, for unit-weight instances, we obtain an algorithm parameterized in the somewhat common FPT parameter (e.g.\ \cite{uhlmann2013two,balaban2025computing}) \hypertarget{paramFES}{\defn{feedback edge set number}}  $\paramFES=m-n+1$, i.e., the number of edges whose removal yields an acyclic subgraph.
\begin{corollary}\label{cor:exFES}
    For unit-weight instances, we yield an FPT algorithm w.r.t.\ $(\paramFES, \paramAP)$ by running the above $(\paramRE, \paramAP)$-algorithm for all $1 \leq \paramRE \leq \paramFES$.
\end{corollary}

In case of basic instances, Kobayashi's algorithm~\cite{Kobayashi2018} becomes feasible. It has the benefit of a pure kernelization, namely that its running time $\bigO(\paramRE^{2\paramRE+3}\alpha^{\paramRE+1}(\alpha+1)^{\paramRE+1} + nm)$ for the decision problem is a \emph{sum} of a parameter-dependent superpolynomial and an input polynomial.
Yet, our algorithm achieves a stronger bound for the former.
Hence, we can simply use the kernelization step of~\cite{Kobayashi2018} (see \Cref{section:pre}) and then use our stronger case-distinction on $|N|$. %
Checking a solution in the kernel requires at most $\bigO(\paramRE^3\alpha^2)$ time instead of our previous $\bigO(\paramRE m)$ time---at most $\paramRE$ BFSs in a graph with at most $\paramRE^2\alpha^2$ edges:

\begin{corollary}\label{cor:exBasic}
    For basic instances, 
    there is a kernelization-based FPT algorithm w.r.t.\ $(\paramRE,\alpha)$ that decides the problem (incl.\ a witness) in $\bigO\big( \paramRE^{2\paramRE+3} \alpha^{2\paramRE+2}+nm \big)$~time.
\end{corollary}

\section{Inclusion Algorithm}\label{section:alg_inclusion}

Now, we consider the natural parameterization $W\coloneqq w(H)$ and its derivatives. Thereby, we leverage the \hypertarget{paramNTEiS}{number} $\paramNTEiS \coloneqq |N \cap F|\leq |F|$ of \defn{nontrivial edges in the spanner}. 
Clearly, by definition, $\paramNTEiS\leq W$.
An FPT algorithm w.r.t.\ $\paramNTEiS$ thus induces one w.r.t.\ $W$ (cf.\ \Cref{cor:W} below).
First, we show that parameterizing only with $W$ (or by $\paramNTEiS$) will not yield an FPT algorithm, even for a constant stretch~$\alpha$.

\begin{theorem}\label{thm:W2hard}
    The \SPANNER problem parameterized in $W$ is \Wtwo-hard, even for unit-length instances on bipartite graphs with constant $\alpha$ and three distinct~weights.
\end{theorem}
\begin{proof}
 We reduce from the \NP-complete \textsc{HittingSet} problem~\cite{Karp1972}: Given a universe~$\universe$, a family of sets $\family=\{S_1,S_2,...\}$ ($S_i\subseteq\universe$ for all $S_i \in \family$), and a number $k \in \mathbb{N}$, we ask whether there is a subset $X \subseteq \universe$ with $|X| \leq k$ such that $X \cap S_i \neq \emptyset$ for all $S_i \in \family$. The problem is \Wtwo-complete w.r.t.~$k$~\cite{cygan2015parameterized}.
    From a given \textsc{HittingSet} instance $I^*$, we construct a \SPANNER instance $I=(G=(V,E),w,\ell=\oneFunc,\alpha=3,W=k)$ and parameterize by $W=k$ such that $I^*$ is a yes-instance if and only if $I$ is a yes-instance (see Fig.~\ref{fig:HShard} for an example):
    For every $S_i\in \family$, we construct two nodes $u_i, v_i \in V$ connected by an edge $\{u_i, v_i\} \in E$ with weight $W+1$. Denote these edges $E_{\family}$. 
    For every $r \in \universe$, we construct two nodes $r_1, r_2 \in V$ connected by an edge $\{r_1, r_2\} \in E$ with weight~$1$. Denote these edges $E_{\universe}$.
    We also connect these node pairs via path of length three, i.e., 
    we add nodes $r',r'' \in V$ and edges $\{r_1,r'\}$, $\{r',r''\}$, and $\{r'',r_2\}$ with weight~$0$.
    Lastly, for each $S_i \in \family$ and each $r \in S_i$, we add edges $\{u_i, r_1\},\{r_2, v_i\} \in E$ of weight~$0$.
    Clearly, the constructed graph $G$ is bipartite.
    All weight-$0$ edges are trivial. Thus all terminal pairs in $E_{\universe} \subseteq K$ are settled as well.
    The critical terminal pairs are thus $C=E_{\family}=\{\{u_i,v_i\}:1\leq i\leq |\family|\}$.
    
    Assume that $I^*$ is a yes-instance with witness $X \subseteq \universe$.
    Our spanner $H=(V, F)$ is constructed as follows: add all weight-$0$ edges. 
    Then for every $r \in X$, add the edge $\{r_1, r_2\}$ to the spanner.
    Since $X$ is a feasible solution, it contains an element $r^{(i)}\in S_i$ for each $S_i$.
    Thus, for every $\{u_i, v_i\} \in E_{\family}=C$, $H$ contains a settling path $\{ \{u_i,r^{(i)}_1\}, \{r^{(i)}_1, r^{(i)}_2\}, \{r^{(i)}_2,v_i\} \}$ of length $3$. So, for stretch $\alpha=3$, $H$ is a feasible spanner with weight $w(H)\leq k$ and $\mu\leq k$.

    Conversely, assume that $I$ is a yes-instance with witness~$H$. Consider any $\{u_i, v_i\} \in E_{\family}=C$.
    Since $w(H)\leq W$, $H$ cannot contain edge $\{u_i,v_i\}$ and the paths induced by only trivial edges are too long (length 5). Thus, for every $\{u_i,v_i\}$, there has to exist some $r^{(i)}\in S_i$ such that $H$ contains the edge $\{r^{(i)}_1, r^{(i)}_2\}$ of weight~$1$, for a total of at most $W$ edges. Then, $X=\{r^{(i)}:1\leq i\leq |\family|\}$ is a feasible solution for $I$ with $|X|\leq W=k$.
\qed \end{proof}

\begin{figure}[tb]
    \centering
    \resizebox{9cm}{!}{\includeTikzFigure{HSproof}}
    \caption{Example for \Cref{thm:W2hard}. Black, blue, and red edges have weight $0$, $1$, and $W+1$, respectively.
    Thereby, $I^*=(\universe=\{a,b,c,d,e\}, \family=\{\{a,b\}, \{b,c,e\}, \{c,d\}, \{d,e\}\} )$.}
    \label{fig:HShard}
\end{figure}

Hence, achieving FPT time requires considering additional parameters.
The \hypertarget{paramLNS}{\defn{local neighborhood size}} $\paramLNS \coloneqq \max_{\{u,v\} \in C} |\bigcup_{P \in \pathSet{uv}} V(P)|$ is the maximum number of nodes visited by settling paths for any critical terminal pair.
It is commonly used to improve the analysis of spanner approximation guarantees, see, e.g.,~\cite{wir2025simple,dinitz2011directed,berman2011improved,dinitz2025hopsets}.
We use two parameters that more precisely represent their algorithmic necessity: 
The \hypertarget{paramBS}{\defn{\bundlesize}} $\paramBS\coloneqq\max_{\{u,v\}\in C} |\bigcup_{P\in\pathSet{uv}} P\cap N|$ is the maximum number of \emph{nontrivial edges} in settling paths for any critical terminal pair. Clearly, while $\paramBS\leq\binom{\paramLNS}2$, $\paramLNS$ can be unbounded for constant $\paramBS$, as the latter does not count trivial edges; thus $\paramBS$ is a stronger parameter than $\paramLNS$.
Alternatively, the \hypertarget{paramBW}{\defn{\bundlewidth}} $\paramBW\coloneqq\max_{\{u,v\}\in C} |\pathSet{uv}|$ is the maximum number of settling paths for any critical terminal~pair.
In unit-length graphs, any settling path for a critical terminal pair $\{u,v\}\in C$ has maximum length $\alpha$, and thus trivially 
$\max\{\paramBW,\paramBS,\paramLNS\}\leq \Delta^{\alpha}$ for \hypertarget{maxDeg}{maximum degree}~$\maxDeg$.
Then \Cref{obs:unitlength} follows from the \NP-hardness of the basic \SPANNER problem for constant $\alpha, \maxDeg$~\cite{cai1994spanners,Gomez2023}.

\begin{observation}\label{obs:unitlength}
    The \SPANNER problem is \paraNP-hard if parameterized by the \bundlewidth $\paramBW$ and/or the \bundlesize $\paramBS$ and/or the local neighborhood size~$\paramLNS$; this holds even for basic instances and constant $\alpha$.
\end{observation}

Thus, none of the above parameters $W$, $\paramNTEiS$, $\paramBW$, $\paramBS$, and $\paramLNS$ yield an FPT algorithm on their own. We show, however, that the combined parameters $(\paramNTEiS,\paramBW)$ and $(\paramNTEiS,\paramBS)$---and consequently $(\paramNTEiS,\paramLNS)$, $(W,\paramBW)$, $(W,\paramBS)$, $(W,\paramLNS)$---each are indeed strong enough to allow FPT complexity.

Consider the connection between \bundlewidth $\paramBW$ and stretch factor $\alpha$. It is natural to assume that the former is dependent on (and directly correlated to) the latter. Surprisingly, the proof of~\Cref{thm:W2hard} shows that this is indeed not always the case. In the above construction, we may pick any natural number $\alpha\geq 3$ upper bounded by some polynomial function in the encoding length $\enc{I^*}$ and replace the length-3 paths over nodes $r',r''$ by paths of length~$\alpha$. The \bundlewidth $\paramBW$ of $I$ will not change and the settling paths for the critical terminal pairs retain length~$3$ (thus \bundlesize $\paramBS\leq 3\paramBW$) with $\paramAP=1$. We thus have:
\begin{observation}\label{obs:largealpha}
    There are instances where the \bundlewidth $\paramBW$, the \bundlesize $\paramBS$ and the \tightness $\paramAP$ are all constant even if $\alpha$ grows with the~instance.
\end{observation}

Consider the relation between the \bundlewidth $\paramBW$ and the \bundlesize $\paramBS$.
Due to trivial edges, $\paramBW$ may be unbounded even for constant $\paramBS$.
Conversely, $\paramBS$ may be unbounded for constant $\paramBW$:
a set $\pathSet{uv}$ may have few paths, but each of those consists of many edges. 
We consider both combined parameterizations $(\paramNTEiS,\paramBW)$ and $(\paramNTEiS,\paramBS)$.
Our algorithm naturally yields a witness (if it exists) in the specified running time.
Both parameterizations are well-suited when only few nontrivial edges must be added to the spanner (e.g.\ in \textsc{SpannerAugmentation} problems), or when feasibility can be decided ``locally'' for each terminal~pair.

\mySubparagraph{FPT algorithm parameterized by $(\paramNTEiS,\paramBW)$.} 
Our algorithm iteratively builds the spanner using a search tree; see Appendix~\ref{apx:Inc} for a recursive pseudo-code implementation.
We say a (partial) solution $H=(V,F)$ is \defn{\green}, if $w(H) \le W$ and $\abs{F \cap \rmE} \le \paramNTEiS$.
Given an instance $I$, we first compute all trivial edges $T \subseteq E$ and all sets $\pathSet{uv}$, $\{u,v\}\in C$, and start with an initial partial solution $H_\mathrm{init}=(V,T)$, answering NO if this subgraph is already not \green.
Clearly, $H_\mathrm{init}$ may leave many critical terminal pairs unsettled.
The algorithm now recursively traverses an implicit search tree, where each search node $\nu$ represents a subgraph $H(\nu)$. The root of our search tree  represents $H_\mathrm{init}$.
At any search node $\nu$, we consider the set of critical terminal pairs $U \subseteq C$ not settled by $H(\nu)$ and select any one $\{u,v\} \in U$.
We create a new search node $\nu'$ for every settling path $P \in \pathSet{uv}$, provided that $H(\nu')\coloneqq(V,E(H(\nu))\cup P)$ is \green.
If this process generates a search node $\nu^*$ where $H(\nu^*)$ settles all terminal pairs $C$, we answer YES; otherwise, we answer NO.

\begin{theorem}\label{thm:f-bw}
    The above algorithm (see also Appendix~\ref{apx:Inc}) is an FPT algorithm 
    that decides the \SPANNER problem parameterized by the maximum number of nontrivial edges in the spanner $\paramNTEiS$ and the \bundlewidth $\paramBW$.
    It runs in $\bigO\big(\paramBW^{\paramNTEiS} \cdot \APSP + \APbSP\big) = \bigO\big(\paramBW^{\paramNTEiS} \cdot \APSP + b \cdot (n^3 \log n + mn^2) \big)$ time, including finding a witness for yes-instances.
\end{theorem}
\begin{proof}
    Clearly, the algorithm considers the whole relevant search space. It remains to argue the running time. 
    For each $\{u,v\}\in C$,~computing $\pathSet{uv}$ can be done via an AP$k$SP-computation with $k=\paramBW$.
    This dominates the computation of the trivial edges $T$ (cf.\ \Cref{section:pre}).

    Our search tree's branching factor is limited by $\paramBW$.
    Further, we only add paths to a partial solution for currently unsettled terminal pairs.
    Whenever we add a path $P$ to our current solution (i.e., in each recursive step), 
    the subsequent solution contains at least one more nontrivial edge. Thus, the depth of our search tree is bounded by $\paramNTEiS$. %
    Overall, we traverse at most $\paramBW^{\paramNTEiS}$ search nodes.

    In each search node, we have to consider the unsettled terminal pairs $U\subseteq C$. These can be found via a single APSP-computation. Since $\APSP\in\Omega(n^2)$, we have $n^2 \paramBW \leq \paramBW^{\paramNTEiS} \cdot \APSP$ and the initialization together with the search tree traversal yields the statement's running time.
\qed \end{proof}

\mySubparagraph{FPT algorithm parameterized by $(\paramNTEiS,\paramBS)$.}
A slight modification of the above algorithm gives us the following result.

\begin{theorem}\label{thm:f-bs}
    The \SPANNER problem parameterized by 
    the maximum number of nontrivial edges $\paramNTEiS$ and the \bundlesize $\paramBS$ can be decided in $\bigO\big(\paramBS^{\paramNTEiS} \cdot \APSP + m^3+m^2 n \log n\big)$ time (including finding a witness if it exists).
\end{theorem}
\begin{proof}
    Consider the algorithm of \Cref{thm:f-bw}; we modify the strategy to yield new search nodes. Let $H$ be the current partial solution and $\{u,v\}\in U \subseteq C$ an unsettled terminal pair. The set $Q_{uv}\coloneqq\bigcup_{P\in\pathSet{uv}} P\cap N$, containing all nontrivial edges in settling paths for $\{u,v\}$, has size $|Q_{uv}|\leq\paramBS$. We know that at least one of its elements needs to be contained in the final solution, so we can generate new search nodes for each \green $H'=H+e$, $e\in Q_{uv}$. We have branching factor at most $\paramBS$ and a search tree depth of at most $\paramNTEiS$.

    For each $\{u,v\} \in C\subseteq E$ and every $\{x,y\} \in N\subseteq E$, we can use Suurballe's  algorithm~\cite{Suurballe1974} to determine an $ux$-path $P_1$ and a $yv$-path $P_2$ in $\bigO(m + n \log n)$ time, where $P_1$ and $P_2$ are node-disjoint and the concatenation $\bar{P}_{xy} = P_1 + \{x,y\} + P_2$ is a shortest simple $uv$-path that traverses the edge $\{x,y\}$ (in that direction). 
    Analogously, let $\bar{P}_{yx}$ be the path when traversing $\{x,y\}$ in the other direction.
    Thus, $\{x,y\} \in Q_{uv}$ if and only if $\min\{\ell(\bar{P}_{xy}),\ell(\bar{P}_{yx})\} \leq \alpha \cdot \dist{G}{u}{v}$. 
    Computing all $Q_{uv}$ takes $\bigO(m^2(m + n \log n))$ time and the claim follows. %
\qed \end{proof}

\begin{corollary}\label{cor:lns}
    The above algorithm is an FPT algorithm w.r.t.\ $(\paramNTEiS,\paramLNS)$, where $\paramLNS$ is the local neighborhood size, with running time $\bigO\big(\paramLNS^{2\paramNTEiS} \cdot \APSP + m^3+m^2 n \log n\big)$, since $\paramBS\leq \paramLNS^2$.
\end{corollary}

Recall from~\Cref{obs:unitlength}, that in unit-length graphs with maximum degree $\maxDeg$ we have bounded $\paramBW \leq \maxDeg^\alpha$.
In such graphs APSP only requires $\bigO(nm)$ time by $n$-many BFSs. 
Similarly, computing the $\paramBW$-shortest-paths between every adjacent node pair only takes $\bigO(b m^2)$ time.
From \Cref{thm:f-bw} it thus follows:
\begin{corollary}\label{cor:ulengthfpt}
    The unit-length (and thus also basic) \SPANNER problem can be solved in $\bigO \left( \maxDeg^{\alpha \paramNTEiS} \cdot nm+\maxDeg^\alpha m^2 \right)$ FPT time when parameterized by $(\paramNTEiS,\alpha,\maxDeg)$.
\end{corollary}

As noted initially, the weight limit $W$ of the spanner itself can be a sensible parameter if the instance contains edges $e \in E$ with $w(e)=0$ (i.e., in the decoupled or unit-length setting). Since $\paramNTEiS\leq W$, we achieve an FPT algorithm using the natural parameterization~$W$\rlap{:}
\begin{corollary}\label{cor:W}
    In 
    Statements \ref{thm:f-bw}--\ref{cor:ulengthfpt},
    we can substitute the parameter $\paramNTEiS$ by $W$.
\end{corollary}

\bibliographystyle{splncs04}
\bibliography{literature}

\clearpage
\appendix

\section*{Appendix}
\section{Omitted Proof}\label{apx:MissingProofs}

\setcounter{theorem}{1}
\begin{lemma}
    Let $G=(V,E)$ be an $n$-node graph with $|V(G[V'])| \geq |E(G[V'])|$ for all $V' \subseteq V$, where $G[V'] \subseteq G$ is the node-induced subgraph. Then $G$ has an independent set of size at least $\lfloor \frac{n}{2}\rfloor$, which can be found in $\bigO(n)$ time.
\end{lemma}
\begin{proof}%
    For a disconnected graph $G$, if the statement holds for each component, it also holds for full $G$. Assume connected $G$. 
    For $n\leq 3$ picking any single node suffices; for $n=4$, there always is a pair of nodes $u,v$ with $\{u,v\} \notin E$.
    For $n> 4$, $n-1 \le |E| \le n$ by connectivity and the assumption $|E| \leq |V| = n$.
    For $|E| = n-1$, $G$ is a tree and thus bipartite. Taking the larger of the two partition sets yields an independent set of size at least $\lceil \frac{n}{2} \rceil$.
    For $|E| = n$, the graph contains exactly once cycle $C$.
    By removing any one node $v\in V(C)$ from $G$ (including its at least 2 incident edges), we may apply the previous arguments to the remaining graph on $n-1$ nodes ($n\geq4$), giving an independent set of size at least $\lceil \frac{n-1}{2} \rceil  = \lfloor \frac{n}{2}\rfloor$. In either case, a simple depth-first-search traversal of the graph suffices to identify an independent set of the specified size. \qed
\end{proof}

\clearpage
\section{Full Exclusion Algorithm of \Cref{section:alg_exclusion}}\label{apx:Ex}
\begin{algorithm}[h]
\DontPrintSemicolon
    \caption{FPT algorithm parameterized by $(\paramRW,\paramAP)$ for instance $(G=(V,E),w,\ell,\alpha,W)$}
    \label{alg:exclusion}
    \SetKwFunction{findWitness}{findWitness}
    \SetKwFunction{main}{main}
    \Fn{\main{}}{
    identify nontrivial edges $N\subseteq E$\\
    \eIf{$|N| \leq \paramAP^2 \paramRW^2$}{
        \ForEach(\Comment*[f]{enumerate all possibilities}){$\noF \subseteq N$ with $|\noF| \leq \paramRW$}{%
            \If{$w(\noF)\geq\paramRW$ and $H = (V, E \setminus \noF)$ is a feasible spanner}{
                \Return YES \Comment*[r]{$H$ is a feasible spanner of suitable weight}
            }
        }
        \Return NO
    }{
        \emph{optional:} $H$ $\gets$ \findWitness{}\\
        \Return YES
    }
    }
    \BlankLine
    \Fn{\findWitness{}}{
        $A_{uv}$ $\gets$ tightest alternative path for each $\{u,v\} \in C$\\
        \If{\paramAP=1}{
            construct conflict graph $\auxG\coloneqq (N, \bigcup_{\{u,v\}\in C, e\in A^N_{uv}} \{ \{u,v\}, e \})$\\
            find independent set $\noF$ in $\auxG$ using \Cref{lem:IS}\\
            \Return $(V,E\setminus \noF)$ \Comment*[r]{feasible spanner of suitable weight}
            }
        $L_0 \gets \emptyset$, $\noF \gets \emptyset$\\
        \For{$i = 1, \dots, \paramRW$}{
            $R_i$ $\gets$ any subset of $N \setminus L_{i-1}$ with $|R_i|=\paramAP\paramRW$\\
            $L_i \gets L_{i-1} \cup \{A^N_{uv} \cup \{\{u,v\}\} \mid \{u,v\} \in R_i \cap C\}$
        }
        \For{$i = \paramRW, \dots, 1$}{
            $\{u_i,v_i\}$ $\gets$ any edge in $R_i \setminus \bigcup_{j>i} A^N_{u_jv_j}$\\
            add $\{u_i,v_i\}$ to $\noF$
        }
        \Return $(V,E\setminus\noF)$ \Comment*[r]{feasible spanner of suitable weight}
    }
\end{algorithm}

\clearpage
\section{Full Inclusion Algorithm of \Cref{section:alg_inclusion}}\label{apx:Inc}

\begin{algorithm}[h]
    \caption{FPT algorithm parametrized by $(\paramNTEiS,\paramBW)$ for instance $(G=(V,E),w,\ell,\alpha,W)$}
    \label{alg:inclusionFramework}
    \SetKwFunction{step}{recursivesearch}
    \SetKwFunction{main}{main}
    \SetKw{KwAbstract}{abstract}

    \Fn{\main{}}{
        compute trivial edges $T$ and $\pathSet{uv}$ for every $\{u,v\} \in C$\\
        $H_\mathrm{init}\gets(V,T)$\\
        \If{$H_\mathrm{init}$ is not \green}        {\Return NO}
        \Return \step{$H_\mathrm{init}$}
    }
    \BlankLine
    \Fn{\step{$H$}}{
        $U \gets \{\{u,v\} \in C \mid \dist{H}{u}{v} > \alpha \dist{G}{u}{v}\}$\\
        \If{$U=\emptyset$}{\Return YES\Comment*[r]{$H$ is a feasible spanner of suitable weight}}
        choose any $\{u,v\} \in U$\\
        \ForEach{$P \in \pathSet{uv}$}{\label{line:allSets}
            $H' \gets H + P$\\
            \If{$H'$ is \green}{
                \Return \step{$H'$} 
            }
        }
        \Return NO
    }
\end{algorithm}

\end{document}